\documentclass[12pt,a4paper]{article}
\usepackage[latin1,utf8]{inputenc}
\usepackage{amsmath}
\usepackage{amsthm}
\usepackage{amsfonts}
\usepackage{amssymb}
\usepackage{thmtools}
\usepackage{thm-restate}
\usepackage{hyperref}
\usepackage{cleveref}
\usepackage{graphicx}
\usepackage{nicefrac}
\usepackage{fullpage}
\usepackage{bm}
\usepackage{xcolor}
\usepackage[linesnumbered,ruled]{algorithm2e}
\usepackage{tikz}
\usepackage{enumitem}

\newtheorem{thm}{Theorem}[section]
\newtheorem{prop}[thm]{Proposition}
\newtheorem{lemma}[thm]{Lemma}
\newtheorem{remark}[thm]{Remark}

\newtheorem{claim}[thm]{Claim}
\newtheorem{defi}[thm]{Definition}
\newtheorem{cnst}[thm]{Construction}

\newcommand{\ceil}[1]{\left \lceil #1 \right \rceil}
\newcommand{\Fq}{\mathbb{F}_q}

\newcommand{\Fqm}[1]{\mathbb{F}_{q^{#1}}}
\newcommand{\cC}{\mathcal{C}}

\newcommand{\ourdel}{1/54}
\newcommand{\ourate}{(1-54\cdot \delta)/1216}

\newcommand{\cCin}{\mathcal{C}_{\textup{in}}}
\newcommand{\cCout}{\mathcal{C}_{\textup{out}}}
\newcommand{\rRin}{\mathcal{R}_{\textup{in}}}
\newcommand{\rRout}{\mathcal{R}_{\textup{out}}}
\newcommand{\delin}{\delta_{\textup{in}}}
\newcommand{\delout}{\delta_{\textup{out}}}

\newcommand{\epsout}{\varepsilon_{\textup{out}}}
\newcommand{\outbuf}{0^{(\textup{out})}}
\newcommand{\inbuf}{0^{(\textup{in})}}

\newcommand\blfootnote[1]{%
  \begingroup
  \renewcommand\thefootnote{}\footnote{#1}%
  \addtocounter{footnote}{-1}%
  \endgroup
}

\DeclareMathOperator{\poly}{poly}
\newcommand{\F}{\mathbb{F}}
\newcommand{\cR}{\mathcal{R}}
\newcommand{\ed}{\text{ED}}

\begin{document}
	
	\title{Explicit and Efficient Constructions of linear Codes Against Adversarial Insertions and Deletions}
\author{
	Roni Con\thanks{Blavatnik School of Computer Science, Tel Aviv University, Tel Aviv, Israel. Email: roni.con93@gmail.com} \and 
	Amir Shpilka \thanks{Blavatnik School of Computer Science, Tel Aviv University, Tel Aviv, Israel. Email: shpilka@tauex.tau.ac.il. The research leading to these results has received funding from the Israel Science Foundation (grant number 514/20) and from the Len Blavatnik and the Blavatnik Family Foundation.} \and
	Itzhak Tamo \thanks{Department of EE-Systems, Tel Aviv University, Tel Aviv, Israel. Email: zactamo@gmail.com.}
}
\date{}
\maketitle
\begin{abstract}
\blfootnote{The work of Itzhak Tamo and Roni Con was partially supported by the European Research Council (ERC grant number 852953) and by the Israel Science Foundation (ISF grant number 1030/15).}
\sloppy In this work, we  study linear error-correcting codes against adversarial  insertion-deletion (insdel) errors, a topic that has recently gained a lot of attention.
	
    We construct linear codes over $\F_q$, for $q=\poly(1/\varepsilon)$, that can efficiently decode from a  $\delta$ fraction of insdel errors and have rate $(1-4\delta)/8-\varepsilon$. We also show that by allowing codes over $\F_{q^2}$ that are linear over $\F_q$, we can improve the rate to $(1-\delta)/4-\varepsilon$ while not sacrificing efficiency. Using this latter result, we construct fully linear codes over $\F_2$ that can efficiently correct up to $\delta < \ourdel$ fraction of deletions and have rate $R = \ourate$. 
    Cheng, Guruswami, Haeupler, and Li \cite{cheng2020efficient} constructed codes with (extremely small) rates bounded away from zero that can correct up to a $\delta < 1/400$ fraction of insdel errors. They also  posed the problem of constructing linear codes that get close to the \emph{half-Singleton bound} (proved in \cite{cheng2020efficient}) over small fields. Thus, our results significantly improve their construction and get much closer to the bound.
    
\end{abstract}
\thispagestyle{empty}
\newpage
\section{Introduction}
	Error-correcting codes are among the most widely used tools and objects of study in information theory and theoretical computer science. The most common model of corruption that is studied in the TCS literature is  that of errors or erasures. The model in which each symbol of the transmitted word is either replaced with a different symbol from the alphabet (an error) or with a `?' (an erasure). The theory of such codes began with the seminal work of Shannon, \cite{shannon1948mathematical}, who studied random errors and erasures and the work of Hamming \cite{hamming1950error} who studied the adversarial model for errors and erasures. These models are mostly well understood, and today we know efficiently encodable and decodable codes that are optimal for Shannon's model of random errors. For adversarial errors, we have optimal codes over large alphabets and  good codes (codes of constant relative rate and relative distance) for every constant sized alphabet. 
	
	Another important model that has been considered ever since Shannon's work is that of \emph{synchronization} errors. These are errors that affect the length of the received word. The most common model for studying synchronization errors is the insertion-deletion model  (insdel for short): an insertion error is when a new symbol is inserted between two symbols of the transmitted word. A deletion is when a symbol is removed from the transmitted word. For example, over the binary alphabet, when $100110$ is transmitted, we may receive the word $1101100$, which is obtained from two insertions ($1$ at the beginning and $0$ at the end) and one deletion (one of the $0$'s at the beginning of the transmitted word).  
	Observe that compared to the more common error model, if an adversary wishes to \emph{change} a symbol, then the cost is that of two operations - first deleting the symbol and then inserting a new one instead.
	
	Insdel errors appear in diverse settings such as optical recording, semiconductor devices, integrated circuits, and synchronous digital communication networks. Another important example is the trace reconstruction problem, which has applications in computational biology and DNA-based storage systems \cite{bornholt2016dna,yazdi2017portable,heckel2019characterization}. See the surveys \cite{mitzenmacher2009survey,mercier2010survey} for a good picture of the problems and applications of error-correcting codes for the insdel model (insdel codes for short). 


	Unlike the (mostly) well-understood error models of Shannon and Hamming for random and worst-case errors, respectively, there are many simple questions regarding the insdel model that are widely open. In fact, even the basic question of what is the \emph{capacity} of the binary deletion channel with error probability $p$ (i.e., the channel in which each symbol of the transmitted message is deleted with probability $p$) is still  open. See the surveys \cite{mitzenmacher2009survey,mercier2010survey} as well as \cite{mitzenmacher2006simple,dalai2011new,cheraghchi2018capacity,con2020explicit} for an overview of results regarding the capacity of the binary deletion channel.
	Due to the importance of the insdel model and our lack of understanding of some basic problems concerning it, the model has attracted many researchers in recent years \cite{haeupler2017synchronization,brakensiek2017efficient,guruswami2017deletion,cheng2018deterministic,haeupler2019optimal,cheng2020efficient,guruswami2021explicit}. However, even the basic question of whether there exist good \emph{linear} codes, over small alphabets,\footnote{Over large alphabets this is easy to achieve, see discussion in \Cref{sec:proofidea}.} for the insdel model was unknown until the recent work of Cheng, Guruswami, Haeupler, and Li \cite{cheng2020efficient}. 
	
	Linear codes are desirable for many reasons: they have a compact representation (they are determined by their generating matrix), they are efficiently encodable, in some settings, we even have linear codes with  linear encoding and decoding time, and often they are simpler to analyze. In \cite{abdel2007linear}, it was shown that linear codes that can correct even one deletion, have a rate at most $1/2$, which is achieved by a trivial repetition code. More generally, in  \cite{cheng2020efficient}, it was shown that codes that can decode from a $\delta$ fraction of insdel errors cannot have rate larger than $(1-\delta)/2 + o(1)$ where the $o(1)$ term goes to zero as the block length tends to infinity. This bound is called the ``half-Singleton bound,'' and it is in sharp contrast to the fact that nonlinear insdel codes, or even \emph{affine} codes (codes  that form an affine space) can achieve rate close to $1$ while still being able to decode from a constant fraction of insdel errors \cite{cheng2020efficient}.	
	While previous work mistakenly claimed that there could be no good binary linear insdel codes (i.e., codes of a constant rate that can handle a constant fraction of insertions and deletions over the alphabet $\{0,1\}$),  \cite{cheng2020efficient} proved that there are binary linear codes of rate $1/2-\varepsilon$ that can decode from $\Omega\left(\varepsilon\log\varepsilon^{-1}\right)$, which is optimal up to the $\log\varepsilon^{-1}$ factor. They also proved that over fields of size $\exp(1/\varepsilon)$ there exist linear codes of rate  $(1-\delta)/2-\varepsilon$  that can decode a $\delta$ fraction of insdel errors. In addition, \cite{cheng2020efficient} gave explicit constructions of linear insdel codes that can decode from $\delta < 1/400$ fraction of insdel errors, however, the rate of the codes in their construction is extremely low.
	This led them to pose the problem of achieving better rate-(edit-)distance trade-offs.
	
\subsection{Basic definitions and notation}
	
	For an integer $k$, we denote $[k]=\{1,2,\ldots,k\}$. 
	Throughout this paper, $\log(x)$ refers to the base-$2$ logarithm. For a prime power $q$, we denote with $\F_q$ the field of size $q$.

	We denote the $i$th symbol of a string $s$ (or of a vector $v$) as  $s_i$ (equivalently $v_i$). Throughout this paper, we shall move freely between representations of vectors as strings and vice versa. Namely, we shall view each vector $v=(v_1, \ldots, v_n)\in \Fq^n$ also as a string by concatenating all the symbols of the vector into one string, i.e., $(v_1, \ldots, v_n) \leftrightarrow v_1 \circ v_2 \circ \ldots \circ v_n$. Thus, if we say that $s$ is a subsequence of some vector $v$, we mean that we view $v$ as a string and $s$ is a subsequence of that string. 	A	\emph{run} $r$ in a string $s$ is a single-symbol substring of $s$ such that the symbol before the run and the symbol after the run are different from the symbol of the run. 
	
	An error correcting code of block length $n$  over an alphabet $\Sigma$ is a subset $\cC\subseteq \Sigma^n$. The rate of $\cC$ is $\frac{\log|\cC|}{n\log|\Sigma|}$, which  captures the amount of information encoded in every symbol of a codeword.	
	A linear code over a field $\F$ is a linear subspace $\cC\subseteq \F^n$. The rate of a linear code  $\cC$ of block length $n$ is $\cR=\dim(\cC)/n$.  Every linear code of dimension $k$ can be  described as the image of a linear  map, which, abusing notation, we also denote with $\cC$, i.e., $\cC : \F^k \rightarrow \F^n$. Equivalently, a linear code $\cC$ can be defined by a \emph{parity check matrix} $H$ such that $x\in\cC$ if and only if $Hx=0$. The minimal distance of $\cC$ with respect to a metric $d(\cdot,\cdot)$ is defined as $\text{dist}_{\cC}:= \min_{v\neq u \in \cC}{d(v,u)}$. When $\cC\subseteq \F_q^n$ has dimension $k$ and minimal distance $d$ we say that it is an $[n,k,d]_q$ code, or simply an $[n,k]_q$ code.
	Naturally, we would like the rate to be as large as possible, but there is an inherent tension between the rate of the code and the minimal distance (or the number of errors that a code can decode from).
	In this work, we focus on codes against insertions and deletions. 
	\begin{defi}
		Let $s$ be a string over the alphabet $\Sigma$. The operation in which we remove a symbol from $s$ is called a \emph{deletion} and the operation in which we place a new symbol from $\Sigma$ between two consecutive symbols in $s$, in the beginning, or at the end of $s$, is called an \emph{insertion}. 
		
		A \emph{substring} of $s$ is a string obtained by taking consecutive symbols from $s$.
		A \emph{subsequence} of $s$ is a string obtained by removing some (possibly none) of the symbols in $s$. 
	\end{defi}

	The relevant metric for such codes is the edit-distance that we define next.
	
	\begin{defi}
	
		Let $s,s'$ be strings over the alphabet $\Sigma$. 
		A \emph{longest common subsequence} between $s$ and $s'$, is a subsequence $s_\textup{sub}$ of both $s$ and $s'$, of maximal length. We denote by $ \textup{LCS}(s,s')$ the length of a longest common subsequence.\footnote{Note that a longest common subsequence may not be unique as there can be a number of subsequences of maximal length. For example in the strings $s=(1,0)$ and $s'=(0,1)$.} 
		
		The \emph{edit distance} between $s$ and $s'$, denoted by $\ed(s,s')$, is the minimal number of insertions and deletions needed in order to turn $s$ into $s'$. One can verify that this measure  indeed defines a metric (distance function). 
	\end{defi}

	\begin{lemma}[See e.g.  Lemma 12.1 in \cite{crochemore2003jewels}]\label{lem:lcs}
		It holds that $\textup{ED}(s,s') = \left|s\right| + \left|s' \right| - 2 \textup{LCS}(s,s') $.
	\end{lemma}
	
	\subsection{Previous results}\label{sec:prev-results}
	
	Codes against synchronization errors were studies since the 1950s. We refer the reader \cite{mitzenmacher2009survey,mercier2010survey} for a detailed account of early work.
	
	The field received a serious boost with the breakthrough result of Haeupler and Shahrasbi \cite{haeupler2017synchronization}.  They introduced the notion of \emph{synchronization strings} (\Cref{def:sync-str}) and used it to give optimal constructions of (non-linear) insdel codes over fixed alphabets.  Specifically, for every $\varepsilon>0$ and $\delta \in (0,1)$ they constructed a code of rate $1-\delta - \varepsilon$ that can efficiently correct a $\delta$ fraction of insdel errors, over an alphabet of size $\left| \Sigma \right| = O_{\varepsilon}(1)$ (\Cref{thm:hs-code}).
	
	Linear codes against worst-case insdel errors were recently studied by
	Cheng, Guruswami, Haeupler, and Li \cite{cheng2020efficient}. Correcting an error in a preceding work, they proved that there are good linear codes against insdel errors. 
	\begin{thm}[Theorem 4.2 in \cite{cheng2020efficient}] \label{thm:random-code}
		For any $\delta > 0$ and prime power $q$, there exists a family of linear codes over $\Fq$ that can correct up to $\delta n$ insertions and deletions, with rate $(1-\delta)/2 - h(\delta)/\log_2 (q)$.
	\end{thm} 
	The proof of \Cref{thm:random-code} uses the probabilistic method, showing that, with high probability, a random linear map generates such code. Complementing their result, they proved that their construction is  almost tight. Specifically, they provided the following upper bound, which they call ``half-Singleton bound,'' that holds over any field.
	\begin{thm}[Half-Singleton bound: Corollary 5.1 in \cite{cheng2020efficient}]
		Every linear insdel code which is capable of correcting a $\delta$ fraction of deletions has rate at most $(1-\delta)/2 + o(1)$.
	\end{thm}
	
    In a recent work, Guruswami, He, and, Li  \cite{guruswami2021zero} strengthened this result for \emph{binary} codes and showed that there exists an absolute constant $\delta$ such that any binary code $\cC \subseteq \{0,1\}^n$ (not only linear) that can decode from $(1/2 - \delta)$ fraction of deletions must satisfy $|\cC| \leq 2^{\text{poly} \log n}$. In particular, we cannot hope to decode a fraction of insdel errors arbitrarily close to $1/2$ with codes of positive rate.
    
	As we already mentioned, \cite{cheng2020efficient} constructed explicit linear codes for insdel errors. Their codes have rate $\mathcal{R} < 2^{-80}$ \cite{GuruswamiKuan21}, a linear time encoding algorithm and an $O(n^4)$ time algorithm for decoding a $\delta<1/400$ fraction of insdel errors \cite[Theorem 1.4, Corollary 7.1]{cheng2020efficient}. 
	They left the question of constructing efficient codes with better rates open.

\subsection{Our results}
	In this paper, we improve the results presented in \cite{cheng2020efficient}. We give explicit constructions of codes over small fields that are efficient (namely, have polynomial-time encoding and decoding algorithms) and almost attain the half-singleton bound. Specifically, 

	\begin{restatable}{thm}{FullLinear}
		\label{thm:full-linear}
		For every small enough constant $\varepsilon > 0$, $\delta \in (0,1/4)$  and $q=\poly(1/\varepsilon)$, there is an explicit construction of a linear code over $\F_q$ of rate  $\mathcal{R} > (1 - 4\delta)/8 - \varepsilon$ that can correct from a $\delta$ fraction of adversarial  insdel errors. Furthermore,  the running time  of the decoding algorithm is $O(n^3)$. 
	\end{restatable}

	By relaxing the linearity requirement, we construct ``half-linear'' codes. We say that a code is half-linear when it is defined over the field $\Fqm{2}$ and is linear over $\Fq$. 
	The half-linear codes that we construct can decode from any $\delta$ fraction of insdel errors, and their rate is close to $(1-\delta)/4$. 
	
	\begin{restatable}{thm}{HalfLinear} 
		\label{thm:half-linear}
		For every small enough constant $\varepsilon > 0$, $\delta \in (0,1)$, and $q=\poly(1/\varepsilon)$ there is an explicit construction of a code over $\F_{q^2}$, which is linear over the subfield $\F_q$, that has rate $\mathcal{R} > (1 - \delta)/4 - \varepsilon$ and can correct from $\delta$ fraction of insdel errors. Furthermore, the running time  of the decoding algorithm is $O(n^3)$. 
	\end{restatable}
	
	Using this construction, we obtain linear binary codes against insdel errors.

	\begin{restatable}{thm}{binLinear} 
		\label{thm:bin-linear}
		There exists an explicit linear binary code that can correct from $\delta<\ourdel$ fraction of worst-case deletions in $O(n^3)$ time and has the rate $\mathcal{R} = \ourate$.
	\end{restatable}

        While the algorithm in \Cref{thm:bin-linear} is only guaranteed to decode from deletions, we note that, information-theoretically, the code can also decode from $\ourdel$ fraction of worst-case insdel errors, as the claim implies a lower bound on the edit distance between any two codewords.

		Theorems~\ref{thm:full-linear} and \ref{thm:bin-linear} improve upon the (explicit) constructions of linear codes given in \cite{cheng2020efficient}, which can handle a fraction $\delta<1/400$ of insdel errors and whose rate is $<2^{-80}$.
		We note, however, that \Cref{thm:bin-linear} only gives an efficient decoder against deletions, whereas the algorithm in \cite{cheng2020efficient} decodes from both insertions and deletions.

\subsection{Proof idea}\label{sec:proofidea}
	
	We first observe that it is easy to construct codes against deletions from any code that can correct erasures: simply add indices to the coordinates of each codeword. Specifically, if $\cC$ is a code that can correct from $e$ erasures,  then we can consider the following code
	\[
	\cC ' = \left \lbrace ((1, c_1), \ldots, (n,c_n)) \mid c \in \cC \right \rbrace \;.
	\]
	It is easy to see  that this code can decode from $e$ adversarial deletions - the missing indices indicate the location of the deletions, and therefore  we can treat them as erasures.
	With a slightly more advanced algorithm, this code can also decode from adversarial insertions (for this to work, we need a code that can decode from errors as well).
		This construction has two problems. The first is that it is not linear. The second is that it requires an alphabet of size  $\Omega(n)$. 
	
	The problem of linearity can be solved as follows. Assume  $\cC \subseteq \mathbb{F}_q^n$  is linear. To add indices while preserving linearity we replace $(i,c_i)$ with  $(c_i, i\cdot c_i)$. Observe that the resulting code is linear over $\mathbb{F}_q$, but symbols  of the codeword are in $\mathbb{F}_{q^2}$.  We shall call such codes half-linear codes. To make the code fully linear, we replace each symbol $(c_i, i\cdot c_i)$ with two symbols, $c_i$ and $i\cdot c_i$. The problem is that now, after adversarial deletions, it is unclear which indices ``survived'' and which were deleted or corrupted. To overcome this difficulty, we add small ``buffers'' of zeros between the different indices. That is, the new codeword is $(c_1,1\cdot c_1, 0,0, c_2,2\cdot c_2,0,0,c_3,\ldots)$. 
	Note that we still need a large alphabet to have $n$ different field elements that can serve as indices. 
	
		To reduce the alphabet size, we use synchronization strings instead of field elements for the indices. Synchronization strings were defined in the breakthrough work of Haeupler and Shahrasbi \cite{haeupler2017synchronization}.  
	
	\begin{defi}\label{def:sync-str}
		A string $S\in \Sigma^n$ is called an $\varepsilon$-synchronization string if for every $1\leq i<j<k\leq n+1$ it holds that $\ed(S[i,j), S[j,k)]) > (1-\varepsilon)\cdot (k-i)$, where $S[i,j)$ denotes the string $S_i\circ S_{i+1} \circ \cdots \circ S_{j-1}$ and $S_i$ is the $i$th coordinate of $S$.
	\end{defi}
	
	Haeupler and Shahrasbi proved the existence of such strings and gave a polynomial-time randomized algorithm for constructing them. An explicit construction, with improved alphabet size, was given in \cite{DBLP:conf/soda/ChengHLSW19}.
		\begin{thm}[Theorem 1.2 in \cite{DBLP:conf/soda/ChengHLSW19}] \label{thm:det-sync-str}
		For every $n \in \mathbb{N}$ and for every $\varepsilon\in (0,1)$, there is a polynomial time (in $n$) deterministic construction of an $\varepsilon$-synchronization string, of length $n$, over an alphabet of size $O(\varepsilon^{-2})$.
	\end{thm} 

	In \cite{haeupler2017synchronization} Haeupler and Shahrasbi showed that synchronization strings could be used instead of indices. Specifically, they proved that if $\cC$ can decode from $d$ hamming errors and $e$ erasures, for  $2d+e< \delta$, and $S=(S_1S_2\ldots S_n)$ is an $\varepsilon$-synchronization string, then the code
	\begin{equation}
	\cC^{\textup{ID}} := \left \lbrace  ((S_1, c_1), \ldots, (S_n, c_n)) \mid c \in \cC \right \rbrace \;, 
	 \label{eq:CID}
	\end{equation}
	can decode from $(\delta - O(\sqrt{\varepsilon}))n$ insdel errors. 
			\begin{thm}[\cite{haeupler2021synchronization}] \label{thm:hs-code}
		Let $\delta, \varepsilon \in (0,1)$ and let $S$ be an $\varepsilon$-synchronization string.
		Let $\cC$ be a code that can decode, in time $T(n)$, from $d$ hamming errors and $e$ erasures, where $2d + e< \delta n$. Then, the code $\cC^{\textup{ID}} := \left \lbrace  ((S_1, c_1), \ldots, (S_n, c_n)) \mid c \in \cC \right \rbrace $ can decode from $(\delta - 12\sqrt{\varepsilon}) n$ insdel errors in time $O(n^2/\sqrt{\varepsilon}) + T(n)$. 
	\end{thm}
	We note that this code is not linear, even when $\cC$ is a linear code, as the synchronization string $S$ is fixed. However, as outlined above, we can tweak this construction to make the code linear while still maintaining its decoding property. We combine this idea with an algebraic geometry code (AG-code) as the base code $\cC$ to obtain our results. We choose these codes as our base codes as they have the best-known rate-distance tradeoff, and in addition, they come with efficient decoding algorithms. 	
Thus, codewords of our code have the form 
		\[
	\cC ' = \left \lbrace (c_1, S_1\cdot c_1, 0,0, \ldots, 0,0, c_n,S_n\cdot c_n) \mid c \in \cC \right \rbrace \;.
	\]
	To further reduce the alphabet to binary, we perform two additional steps. First, we concatenate our code from \Cref{thm:half-linear} with a carefully chosen binary code of fixed length. Then we add \emph{buffers} of zeroes between any two concatenated words. A short buffer between the encodings of $c_i$ and $S_i\cdot c_i$ and a long buffer between the encodings of $S_i\cdot c_i$ and $ c_{i+1}$. The buffers allow our decoding algorithm to correctly identify the encoding of many pairs $(c_i,S_i\cdot c_i)$. Then, by using the synchronization string, $S$, and the decoder of $\cC$, we obtain a decoding algorithm.

	\subsection{Organization of the paper}
	The paper is organized as follows. In \Cref{sec:finite-field-insdel} we construct linear (and half-linear) codes, over small alphabets, that can handle insdel errors, and prove \Cref{thm:half-linear} and \Cref{thm:full-linear}. In \Cref{sec:binary-del}, we give the construction of linear binary codes that can decode from deletions, thus proving  \Cref{thm:bin-linear}.
	
	\section{Linear Insdel Codes over Finite Alphabet via Synchronization Strings}
	\label{sec:finite-field-insdel}
    In this section, we prove Theorems~\ref{thm:full-linear} and~\ref{thm:half-linear}. We follow the strategy outlined in Section~\ref{sec:proofidea}.
    
    As our base code $\cC$, we shall use an AG-code. The well-known construction of \cite{tsfasman1982modular} beats the Gilbert-Varshamov bound\footnote{The Gilbert-Varshamov bound shows what parameters random (linear) codes achieve.} over $\F_q$, for $q \geq 49$.
     Moreover, this code has an efficient decoder that can correct both errors and erasures, almost up to its correction capability \cite{skorobogatov1990decoding,kotter1996fast}. The interested reader is referred to \cite{stichtenoth2009algebraic} for further information on AG-codes and their decoding.
	
	\begin{thm}[\cite{tsfasman1982modular,skorobogatov1990decoding,kotter1996fast}] \label{thm:tvz-ag-code}
		Let $q = p^{2m}$ be a square where $p$ is a prime and $m$ is a positive integer. For every $0<\delta \leq 1 - \frac{1}{\sqrt{q} - 1}$ there exists an explicit linear code $\cC$ over $\Fq$, of minimal distance $\delta$ and rate
		\[
		\mathcal{R} \geq  1 - \frac{1}{\sqrt{q} - 1} - \delta \;.
		\]
		Moreover, there is a decoding algorithm that runs in time $O(n^3)$ and can correct from $d$ hamming errors and $e$ erasures, for $2d + e < \left(\delta - \frac{1}{\sqrt{q} - 1}\right) n$.
	\end{thm}
	
	
	We first prove Theorem~\ref{thm:half-linear} as the proof of its decoding algorithm is easier and then prove Theorem~\ref{thm:full-linear}.

	\subsection{Half-linear insdel codes}
	

	\begin{cnst} \label{const:half-finite}
		Let $\delta \in (0,1)$ and $\varepsilon$ a small constant. Let $p$ be a prime  such that $p = \Theta (\varepsilon^{-2})$ and set $q = p^2 = \Theta (\varepsilon^{-4})$.
		Set $\delta_{\cC} = (1 + \delta + 13\varepsilon)/2$ and let $\cC$ be the code from $\Cref{thm:tvz-ag-code}$,  defined over the finite field $\Fq$, with rate $\mathcal{R}_{\cC} > 1 - \delta_{\cC} - \varepsilon$.
		Let $S = (S_1S_2\ldots S_n)$ be an $\varepsilon^2$-sync string, where $S_i \in \Fq \setminus \{ 0 \}$ for all $i\in [n]$. Let $\text{Enc}_{\cC}:\F_q^k\to\F_q^n$ be the encoding map of $\cC$. We define the code $\cC'$ via the encoding map $\text{Enc}_{\cC'}$: 
		For $a\in\F_q^k$, let $\text{Enc}_{\cC}(a)=c=(c_1,\ldots,c_n)$. Then, 
		\begin{equation}\label{eq:AGhalf}
		\text{Enc}_{\cC'}(a)=\left( (c_1, S_1 \cdot c_1 ), (c_2, S_2 \cdot c_2), \ldots, (c_n, S_n \cdot c_n) \right)\;.    
		\end{equation}
	\end{cnst}
	Namely, $\cC'$ is  the image of $\Fq ^k$ under $\text{Enc}_{\cC'}$. 
	One can easily observe that the rate $\cC'$ is $\mathcal{R}_{\cC'} = \mathcal{R}_{\cC}/2 > (1 - \delta)/4 -   4\varepsilon$, and that the code is linear over $\mathbb{F}_q$.

	\begin{prop}\label{prop:half-analysis}
		Algorithm~\ref{alg:half-lin} runs in time $O(n^3)$ and can decode $\cC'$ (given in \Cref{const:half-finite}),  from $\delta n$ worst case insdel errors.
	\end{prop}

	\begin{figure}
		\begin{algorithm}[H] 
			\SetAlgoLined
			\DontPrintSemicolon
			\LinesNumberedHidden
			\SetNlSty{textbf}{[}{]}
			
			\SetKwInOut{Input}{input}
			\SetKwInOut{Output}{output}
			\SetNlSty{large}{[}{]}
			\Input{A corrupted codeword $y =(e_1, \ldots, e_t)$.}
			\Output{A message $x \in \Fq ^k$.}
			\nlset{1} Set $L$ to be an empty list\label{step:half-first} \;
			\nlset{2} \label{step:half-comp-S-prime}
			\For{$i=1,\ldots, t$}{
				Let $e_i = (a,b)$\;
				\If{$b = 0$}{Go to the next $i$}
				{Add to $L$ the tuple $(b/a, a)$} \label{step:ignore-zero}
			}
			\label{step:half-extract-points}
			\nlset{3}
			{If $L$ is empty, return the zero codeword $c=0$; else} decode $L$ using the decoding algorithm of $C^{\textup{ID}}$ given in \Cref{thm:hs-code}.
			\label{step:half-decode} \;
			\nlset{4} Let $c^{\textup{ID}} = ((S_1, c_1), \ldots, (S_n, c_n))$ be the decoded codeword. Return the codeword $c = ((c_1, S_1c_1), \ldots, (c_n, S_n c_n))$.
			\label{step:half-final}
			\caption{Decode $\cC'$}
			\label{alg:half-lin}
		\end{algorithm}
	\end{figure}

	\begin{proof}
    For $c = ((c_1, S_1c_1), \ldots, (c_n, S_n c_n))\in \cC'$  let $c^{\textup{ID}} = ((S_1, c_1), \ldots, (S_n, c_n)) \in \cC^{\textup{ID}}$. Observe that  $\cC^{\textup{ID}}$ is as in  Equation~\eqref{eq:CID}. To prove the claim we shall interpret insdel errors in $\cC'$ as insdel errors in $\cC^{\textup{ID}}$ and then apply Theorem~\ref{thm:hs-code}.
    
Assume first that the  corrupted codeword is the  zero vector. Then, since the hamming-weight of each nonzero codeword of $\cC'$ is at  	least $\delta_\cC n >\delta n$, the only codeword that would produce this   corrupted codeword from $\delta n$ insdel errors is the   zero codeword, hence, in Step~\ref{step:half-decode} successfully decodes the zero codeword. Next, we assume that  the corrupted codeword is \emph{not} the   zero vector.

		The map $(a,b)\to (b/a,a)$ maps each {nonzero} coordinate of $c\in \cC'$ to the corresponding coordinate of $c^{\textup{ID}}\in \cC^{\textup{ID}}$ and therefore, by applying it coordinate-wise, we can interpret any insdel error to $c$ as an insdel error to $c^{\textup{ID}}$. 
		
		Observe that in addition to the errors introduced by the adversary, in    Step~\ref{step:half-extract-points} of Algorithm~\ref{alg:half-lin} we treat any zero coordinate as a deletion. Since the minimal distance of $\cC$ is $\delta_{\cC}$, a nonzero codeword $c\in \cC'$ has at most $n(1 -\delta_{\cC})$ zero coordinates. Therefore, Step~\ref{step:half-extract-points}  can cause  $(1 -\delta_{\cC}) n$ additional insdel errors.
		In conclusion,
		\[
		\text{ED}(c^{\textup{ID}}, L) \leq (1 - \delta_{\cC} + \delta) n = (\delta_{\cC} - 13\varepsilon) n\;,
		\]
        where the equality follows from the choice of $\delta_{\cC}$ in Construction~\ref{const:half-finite}.
		As $\cC$ can correct from $d$ hamming errors and $e$ erasures, for  $2d + e \leq (\delta_{\cC} - \varepsilon)n$, \Cref{thm:hs-code}  implies that Step~\ref{step:half-decode} succeeds, and the decoder outputs $c^{\textup{ID}}$. Step~\ref{step:half-final} clearly returns  the codeword $c$.
		
		To prove the claim regarding the running time we note that Steps~\ref{step:half-first} and~\ref{step:half-extract-points} take linear time and that by \Cref{thm:tvz-ag-code}, the decoding algorithm of \Cref{thm:hs-code} runs in time $O(n^2/\varepsilon)+O(n^3)=O(n^3)$.
	\end{proof}

		\begin{remark} \label{rem:zeros-in-hlinear}
	    As the proof shows, Step~\ref{step:ignore-zero} of Algorithm~\ref{alg:half-lin} ignores the symbol $(0,0)$. In other words, it treats this symbol as a deletion. Thus, from the point of view of the adversary, there is no need to corrupt the zero symbol.
	\end{remark}
	
		\begin{proof}[Proof of Theorem~\ref{thm:half-linear}]
		    The proof follows immediately from Construction~\ref{const:half-finite} and \Cref{prop:half-analysis}. Indeed, the code described in Construction~\ref{const:half-finite} maps $k$ symbols of $\F_q$ to $n$ symbols of $\F_{q^2}$ and hence its rate is {$k/(2n)$}. As $k$ was chosen so that $\mathcal{R}_{\cC}=k/n>1-\delta_{\cC}-\varepsilon = (1-\delta-15\varepsilon)/2$, we get that 
		     $\mathcal{R}_{\cC'} = \mathcal{R}_{\cC}/2 > (1 - \delta)/4 -   4\varepsilon$. By construction the code is linear over $\mathbb{F}_q$. 
		\end{proof}
	
	\subsection{Full linear insdel codes}
	We next prove \Cref{thm:full-linear}.
	As described in Section~\ref{sec:proofidea}, to get full linear insdel codes we use a similar construction albeit with two significant modifications: First, we ``flatten'' the code, i.e., we expand each symbol $(c_i, S_i\cdot c_i)\in\F_{q^2}$ to two symbols $c_i,S_i \cdot c_i\in\F_q$. Secondly, to protect our codeword from insdel errors, we additionally insert two zeros between every two adjacent pairs. Thus, the corresponding word to  $\left( (c_1, S_1 \cdot c_1 ), (c_2, S_2 \cdot c_2), \ldots, (c_n, S_n \cdot c_n) \right)$ is $\left(c_1, S_1\cdot c_1, 0, 0 , c_2, S_2\cdot c_2, 0, 0, \ldots, c_n, S_n \cdot c_n \right)$. It is clear that in this way we get a linear code. Formally: 
	\begin{cnst} \label{const:lin-finite}
		Let {$\delta \in (0,1/4)$ and $\varepsilon$ a small enough  constant. Set  $\delta_{\cC} = (1 + 4\delta + 13\varepsilon)/2<1$}. Let $p$ be a prime such that $p = \Theta (\varepsilon^{-2})$ and set $q = p^2$.
		Let $\cC$ be the code from $\Cref{thm:tvz-ag-code}$,  defined over the finite field $\Fq$, with minimal distance $\delta_C$ and rate $\mathcal{R}_{\cC} = 1 - \delta_{\cC} - \varepsilon$.
		Let $S = (S_1S_2\ldots S_n)$ be an $\varepsilon^2$-sync string, where $S_i \in \Fq{\backslash\{0\}}$ for all $i\in [n]$
		Let $\text{Enc}_{\cC}:\F_q^k\to\F_q^n$ be the encoding map of $\cC$. We define the code $\cC''$ via the encoding map $\text{Enc}_{\cC''}$: 
		For $a\in\F_q^k$, let $\text{Enc}_{\cC}(a)=c=(c_1,\ldots,c_n)$. Then, 
		\begin{equation}\label{eq:AGlin}
		\text{Enc}_{\cC''}(a)=  \left(c_1, S_1\cdot c_1, 0, 0 , c_2, S_2\cdot c_2, 0, 0, \ldots, c_n, S_n \cdot c_n \right)\;.    
		\end{equation}
	\end{cnst}
	Namely, $\cC''$ is the image of $\Fq ^k$ under $\text{Enc}_{\cC''}$. 
	Clearly, $\cC''\subset\F_q^{4n-2}$ is an $\Fq$ linear space. 

	\begin{prop}\label{prop:lin-analysis}
		Algorithm~\ref{alg:lin} runs in time $O(n^3)$ and can  decode $\cC''$, given in \Cref{const:lin-finite},  from $\delta n$ worst case insdel errors.
	\end{prop}

	\begin{figure}
		\begin{algorithm}[H] 
			\SetAlgoLined
			\DontPrintSemicolon
			\LinesNumberedHidden
			\SetNlSty{textbf}{[}{]}
			
			\SetKwInOut{Input}{input}
			\SetKwInOut{Output}{output}
			\SetNlSty{large}{[}{]}
			\Input{A corrupted codeword $y =(e_1, \ldots, e_t)$.}
			\Output{A message $x \in \Fq ^k$.}
			\nlset{1} Set $L$ to be an empty list  \;
			Write $y$ as 
			\[y = s_1 \circ \bar{0} \circ s_2 \circ \bar{0} \circ \cdots \circ \bar{0} \circ s_m, \]
			where  $0\leq m \leq t$, the $s_i$'s are strings of symbols that do not contain any $0$'s,  and the notation $\bar{0}$ corresponds to  a string of consecutive zeros of any length.\;
			
			\nlset{2} \For{$i=1,\ldots, m$}{
				\If{$|s_i| \neq 2$}{Continue}
				Let $a,b$ be the first and second elements in $s_i$. 
				Add to $L$ the tuple $(b/a, a)$\;
			} \label{step:lin-compute-Sprime}
			\nlset{3}
			{If $L$ is empty, return the zero codeword, $c=0$; else} decode $L$ using the algorithm of $C^{\textup{ID}}$ given in \Cref{thm:hs-code}.
			\label{step:lin-decode}\;
			\nlset{4} Let $c^{\textup{ID}} = ((S_1, c_1), \ldots, (S_n, c_n))$ be the decoded codeword. Return the codeword $c = (c_1, S_1c_1,0,0, \ldots, 0,0, c_n, S_n c_n)$.
			\label{step:lin-final}
			\caption{Decode $\cC''$}
			\label{alg:lin}
		\end{algorithm}
	\end{figure}

	\begin{proof}
		    \sloppy Let $c=(c_1, S_1c_1,0,0, \ldots, 0, 0, c_n, S_n c_n)\in \cC''$ and denote by $c^{\textup{ID}} = ((S_1, c_1), \ldots, (S_1, c_n)) \in \cC^{\textup{ID}}$ the corresponding codeword, where $\cC^{\textup{ID}}$ is as in the proof of \Cref{prop:half-analysis}.
		We will follow the same reasoning as in the proof of \Cref{prop:half-analysis}; translate insdel errors in $\cC''$ to insdel errors in  $\cC^{\textup{ID}}$, and then apply \Cref{thm:hs-code}.

 Assume first that the  corrupted codeword is the  zero vector. Then, since the hamming-weight of each nonzero codeword of $\cC$ is at  	least $\delta_\cC n$ and $S_i\neq 0$ for each $i$, the normalized minimum distance   of $\cC''$ is at least $2\delta_\cC n/(4n-2)>\delta_\cC /2$. On the other hand, 
$$
\delta< \frac{1}{4}<\frac{1+4\delta+13\varepsilon}{4}=\frac{\delta_\cC}{2}. 
$$ Hence, the only codeword that would produce this   corrupted codeword from $ \delta(4n-2)$ insdel errors is the   zero codeword, and Step~\ref{step:lin-decode} successfully decodes the zero codeword. Next, we assume that  the corrupted codeword is \emph{not} the   zero vector.

		Since the minimal distance of $\cC$ is at least $\delta_{\cC}n$, any nonzero $c\in\cC''$ contains at most $n(1 - \delta_{\cC}) $ pairs $c_i, S_ic_i$ that are equal to $0,0$.
		Every such zero pair is interpreted as a deletion in Step \ref{step:lin-compute-Sprime} of Algorithm \ref{alg:lin}. These deletions are in addition to those made by the adversary.
        The adversary, who knows the decoding algorithm, will clearly ignore the zero pairs $c_i,c_iS_i$ for $c_i=0$, and therefore will either   ``ruin'' nonzero pairs by converting them to nonzero blocks (i.e., blocks with no zeros) of  lengths different than $2$, or by constructing erroneous pairs. 

    The most economic way to construct the former is by  inserting (deleting) a symbol to (from)  an existing nonzero pair, respectively. This increases $\ed(c^{\textup{ID}},L)$ by $1$.  
    Also, the adversary can merge, say $b\geq 2$ consecutive blocks, into a single block by deleting the buffers between them. This  ``costs'' $2(b-1)$ deletions that translate to an increase to $\ed(c^{\textup{ID}},L)$ by $b$.
    Hence, on average, each deletion or insertion in a nonzero block of length different than $2$ increases the edit distance by at most $1$. 

    The construction of the latter, i.e.,     an erroneous pair, would clearly cost $2$ insertions between the zeros of a buffer or by a symbol deletion from an existing nonzero pair, followed by a new nonzero symbol insertion. This is clearly  less economical than  ruining nonzero pairs, since in this case, on average, in order to increase $\ed(c^{\textup{ID}},L)$ by $1$, the adversary must perform two edit operations. 

		To conclude, the accounting above indicates that every insdel error made by the adversary increases the edit distance between $L$ and $c^{\textup{ID}}$ by at most one. It follows that  after the adversary performs $\delta \cdot (4 n-2)$ insdel errors (recall that $c\in\F_q^{4n-2}$),  
		\[
		\text{ED}(c^{\textup{ID}},L)\leq (1 - \delta_{\cC}) n + 4\delta n = (\delta_{\cC} -  13\varepsilon)n \;. 
		\]
		Thus, by \Cref{thm:hs-code} and since the code $\cC$ can correct from $d$ hamming errors and $e$ erasures where $2d + e \leq (\delta_{\cC} - \varepsilon)n$, Steps~\ref{step:lin-decode} and \ref{step:lin-final} succeed.
		
		The claim regarding the running time follows exactly as in the proof of \Cref{prop:half-analysis}.
	\end{proof}
	
	We now conclude the proof of \Cref{thm:full-linear}.
	
	\begin{proof}[Proof of \Cref{thm:full-linear}]
	    As before, the proof is immediate from \Cref{const:lin-finite} and \Cref{prop:lin-analysis}. The rate satisfies 
	\[
	\mathcal{R_{\cC''}} =\frac{R_{\cC}}{4} > \frac{1-\delta_{\cC}-\varepsilon}{4} = \frac{1 - 4 \delta-15\varepsilon}{8}\geq \frac{1 - 4 \delta}{8} - 2\varepsilon \;.\qedhere
	\]
	\end{proof}

 \section{Binary Linear Codes}
    \label{sec:binary-del}
In this section we prove \Cref{thm:bin-linear}. To ease the reading, we repeat the statement of the theorem.
\binLinear*

As explained in \Cref{sec:proofidea} our construction concatenates the code of \Cref{thm:half-linear}  with an adequately chosen short binary code and then adds buffers between the encoding of different symbols: short buffers between the encoding of $c_i$ and $S_i\cdot c_i$ and long buffers between the encodings of $S_i\cdot c_i$ and $c_{i+1}$. The specially tailored inner code is a linear binary code that can correct from a small fraction of insdel errors and has the property that, with the  exception of the zero word, no codeword has large runs of zeroes. We shall prove that such codes exist and then construct one greedily.


\subsection{The inner code}\label{sec:inner-code}

The following proposition describes the  properties that our inner code should possess and is proved  using the probabilistic method. As the code has a fixed length, we shall use the brute force algorithm to construct it.

\begin{prop} \label{prop:inner-code} 
    Set $\delin = 1/6$ and $\rho  = 1/17$. There exists $m_0\in \mathbb{N}$ such that for any $m>m_0$,  which is a multiple of $102$,\footnote{We require this  to ensure that  both  $\rho m$ and $\delin m $ are integers, in order to avoid the use of ceilings and floors.} there is 
    a binary linear code $\cCin \subset \{0,1\}^m$ of rate $\rRin = \delin /16$ such that
    \begin{enumerate}
        \item For any two substrings $c_s, c_s'$ of any two distinct codewords $c\neq c'\in \cCin$ such that $|c_s|, |c_s'| \geq (1 - 2\delin + \rho)m$, it holds that LCS$(c_s, c_s') < \min (|c_s|, |c_s'|) -\rho m$. \label{prop:inner-1}
    
        \item Any substring $c_{\textup{sub}}$ of length $\delin m$, of any nonzero codeword $c\in \cCin$ contains at least $\rho m + 1$ ones.\label{prop:inner-2}
    \end{enumerate}
\end{prop}

    Observe that \Cref{prop:inner-code}\eqref{prop:inner-1} implies that $\ed(c_s, c_s')> 2 \rho m$ so in particular we can brute force correct any $\rho m$ insdel errors in $\cCin$ in time $\exp(m)$. 

\begin{proof}
    Let $G \in \mathbb{F}_2 ^{m\times \rRin m}$ be a uniformly chosen random matrix. $G$  will serve as a generator matrix for a linear code $\cC$, i.e., $\cC = \{ Gv \mid v\in \mathbb{F}_2^{\rRin m} \}$.
    We next prove that the probability that $\cC$ does not satisfy any of the  properties in the proposition is  small.

    The proof that 
    \Cref{prop:inner-code}\eqref{prop:inner-1}  holds with  high probability relies on the following simple and intuitive claim given in \cite{cheng2020efficient}.
    \begin{claim}[Claim 4.1 of \cite{cheng2020efficient}]
       Let $\cC$ be a random linear code and let $c\neq c'$ be any two distinct codewords. Fix two sets of indices $\{s_1, \ldots, s_t\}, \{s_1', \ldots, s_t'\} \subset [n]$. Then,
       \[
       \Pr [\forall i\in [t], (c)_{s_i} = (c')_{s_i'}] \leq 2^{-t}\;.
       \]
    \end{claim}
    Let $c\neq c' \in \cC$ be distinct and $c_s$ and $c_s'$ be substrings of $c$ and $c'$, such that  $r = \min (|c_s|, |c_s'|)\geq  (1 - 2\delin + \rho)m$. Let $\{s_1, \ldots, s_{r - \rho m}\}$ and $\{s_1', \ldots, s_{r - \rho m}'\}$ be two sequences of indices.  The claim implies  that,
    \[
    \Pr \left[ \forall i\in [{r - \rho m}] , (c_s)_{s_i} = (c_s')_{s_i'}\right]\leq 2^{-(r - \rho m)} \leq 2^{-(1 - 2\delin)m} \;.
    \]
    By the union bound, the probability that $c_s$ and $c_s'$ share a common subsequence of length $(r-\rho m)$  is at most
    \[
    \binom{r}{r - \rho m} ^2 \cdot 2^{-(1 - 2\delin)m} \leq 2^{ m \cdot \left(  2 h\left(\rho \right) - (1-2\delin)\right) \;, }
    \]
    where we used  $\binom{r}{r - \rho m} = \binom{r}{\rho m} \leq \binom{m}{\rho m}$.
    Now, the number of subsrings of $c$ ($c'$) of length $\geq (1 - 2\delin + \rho)m$ is at most $m^2 \cdot (2\delin - \rho)$ and the number of codewords is $2^{\rRin m}$. Thus, the probability that there exist  $c \neq c'$, and  substrings $c_s$ and $c'_s$ of $c$ and $c'$, respectively, such that $|c_s|,|s'_s|\geq (1 - 2\delin +\rho) m$ and they share a common subsequence of length $r - \rho m$ is at most 
    \[
    2^{ 2m \cdot \rRin} \cdot m^2 \cdot 2^{ m \cdot \left( 2 h\left(\rho \right) - (1-2\delin) \right) } =
    2^{ 2m \cdot \left(\rRin+  h\left( \rho \right) - (1-2\delin)/2 + \frac{O(\log m)}{m}\right) } \;.
    \]
    Thus, as long as 
    \begin{equation} \label{eq:constr-1}
    \rRin + h\left( \rho \right) - (1-2\delin)/2 < 0 \;,
    \end{equation}
    there exists $m_0' \in \mathbb{N}$ such that for every integer $m\geq m_0'$, the probability that \Cref{prop:inner-code}\eqref{prop:inner-1} does not holds is smaller than $1/4$.
%
    
    To prove that \Cref{prop:inner-code}\eqref{prop:inner-2} holds with high probability, consider any $0\neq v \in \mathbb{F}_2 ^{\rRin m}$. As $G$ was chosen uniformly at random, $Gv$ is uniform random vector in $\mathbb{F}_2^m$. 
    The probability that $Gv$  contains a substring of length $\delin m$  that has $\leq \rho m$ ones is at most 
    \[
    m\cdot \sum_{i=0}^{\rho m} \binom{\delin m}{i} 2^{-\delin m} \leq m (\rho m  + 1) \cdot \binom{\delin m}{\rho m} \cdot 2^{-\delin m} \leq 2^{\delin m \left(- 1 +  h\left(\frac{\rho}{\delin}\right) + \frac{O(\log (m))}{m}\right)} \;.
    \]
    Thus, by the union bound, the probability that there exists $v \in \mathbb{F}_2^{\rRin m} \setminus \{0\}$, such that $Gv$ contains a substring of length $\delin m$ with $\leq \rho m$ ones is at most
    \[
        2^{m \left(\rRin - \delin  +  \delin  h\left(\frac{\rho}{\delin}\right) + \frac{O(\log (m))}{m}\right)} \;.
    \]
    Hence, if 
    \begin{equation} \label{eq:constr-2}
    \rRin - \delin  +  \delin  h\left(\frac{\rho}{\delin}\right) < 0 \;
    \end{equation}
    then there exists $m_0'' \in \mathbb{N}$ such that for every integer $m\geq m_0''$, the probability that $\cC$ does not satisfy this property is $\leq 1/4$.
    
    It can be verified that for $\delin = 1/6$, $ \rho = 1/17$, $\rRin = \delin/16$, and $m_0 = \max(m_0', m_0'')$,  inequalities \eqref{eq:constr-1}  and \eqref{eq:constr-2} hold true and therefore the probability that a random code $\cC$ satisfies both properties is at least $ 1/2$ and the proposition follows. 
    \end{proof}

\paragraph{Construction and decoding}
To explicitly construct codes as in \autoref{prop:inner-code} we simply go over all possible linear codes and pick one that satisfies both properties. This requires $\exp(m^2)$ many steps. In our final construction we need $m =O(\log (1/\epsout))$ and hence the cost of constructing the inner code is $\exp(\log^2(1/\epsout))$.

Similarly, we decode from deletions using the following brute force algorithm: Set $L'$ to be an empty list. On input $\tilde{c}$, the algorithm runs over every codeword $c\in \cC$ and  checks if $\tilde{c}$ is a subsequence of $c$. If the answer is yes and $c$ is not in $L'$, then the algorithm adds $c$ to $L'$. 
If $L'$ contains only $c$, then the algorithm returns $c$. Otherwise, it returns $\perp$. Clearly, the running time of this algorithm is $\exp(m)=\poly(1/\epsout)$.



\begin{remark} \label{rem:unique-dec}
   An important observation is that our decoding algorithm cannot output a wrong answer. Indeed, if  $\tilde{c}$  was obtained from $c$ by performing any number of deletions, then $c$ will be one of the codewords in $L'$ (as $\tilde{c}$ is a subsequence of $c$).
\end{remark}

\subsection{Construction of our code}

    Let $\delout>0$ and $\epsout< \delout/1400$ small enough. Let $\cCout \subset \Fqm{2}^n$ be the code given in \Cref{thm:half-linear},  with parameters $\delta=\delout$ and $\varepsilon= \epsout$.  Recall that the rate of $\cCout$ is $\rRout = (1 - \delout)/4 - \epsout$ and the code is defined over the alphabet $\Fqm{2}$ where $q = \text{poly}(1/\epsout)$. Denote $k=\rRout\cdot n$. Let $\cCin: \{0,1\}^{m\cdot \rRin} \rightarrow \{0,1\}^m$ be the code obtained in \Cref{sec:inner-code}, where $m$ is such that $\rRin m = \log(q)$ (we pick $\epsout$ small enough so that $m\geq m_0$ as in \Cref{prop:inner-code}).
    
    \begin{cnst} \label{cnst:binary-lin}
    The encoding works as follows. Given a message $x\in \mathbb{F}_q^k$ we:
    \begin{enumerate}
        \item Encode $x$ using the outer code $\cCout$ to obtain $\sigma = \cCout(x)$. Denote 
        \[
        \sigma = \left((\sigma_1, S_1\cdot \sigma_1), \ldots, (\sigma_n, S_n \cdot \sigma_n) \right)\;.
        \]
        \item Let $\inbuf$ denote a string of $2\delin m$ many zeroes.
        Encode every symbol $(\sigma_i, S_i \cdot \sigma_i)$ using the inner code to obtain $\left(\cCin(\sigma_i), \cCin(S_i \cdot \sigma_i)\right)$ and place the string $\inbuf$   between $\cCin(\sigma_i)$ and $\cCin(S_i \cdot \sigma_i)$. We refer to those $\inbuf$ strings  as \emph{inner buffers}. At the end of this step we have the string
        \[
        \cCin(\sigma_1) \circ \inbuf \circ \cCin(S_1 \cdot \sigma_1) \circ \ldots \circ \cCin(\sigma_n) \circ \inbuf \circ \cCin(S_n \cdot \sigma_n) \;.
        \]
        \item Let $\outbuf$ denote a string of $5\delin m$ many zeroes.
        Place the string $\outbuf$ between every two adjacent symbols of the form $\cCin(S_i \cdot \sigma_i) \circ  \cCin(\sigma_{i+1})$  to get
        \[
        \cCin(\sigma_1) \circ \inbuf \circ \cCin(S_1 \cdot \sigma_1) \circ \outbuf \circ \ldots \circ \cCin(\sigma_n) \circ \inbuf \circ \cCin(S_n \cdot \sigma_n)  \;.
        \] 
        We refer to those $\outbuf$ strings  as \emph{outer buffers}.
        \end{enumerate}
        The encoding of $x$ is the string
        \[\textup{ENC}(x)= \cCin(\sigma_1) \circ \inbuf \circ \cCin(S_1 \cdot \sigma_1) \circ \outbuf \circ \ldots \circ \cCin(\sigma_n) \circ \inbuf \circ \cCin(S_n \cdot \sigma_n) \;.\]
    \end{cnst}
    
    \paragraph{Rate:}
    The length of the codewords is $2mn + 2\delin mn + 5\delin m (n - 1)$ bits. Recalling that $\log(q) = m \cdot \rRin$ we get
    \begin{align} \label{eq:rate-bin}
    \mathcal{R} &= \frac{\log(q^{\rRout n})}{2mn + 2\delin mn + 5\delin m (n-1)} \nonumber \\
    &> \frac{\rRin \rRout}{2 + 7 \delin} \;.
    \end{align}
    
    The decoding algorithm is given in Algorithm \ref{alg:decode}. 
    
    \begin{algorithm}[h]\label{alg:decode}
    \SetKwInOut{Input}{input}
    \SetKwInOut{Output}{output}
    \SetNlSty{textbf}{[}{]}
    \SetAlgoNlRelativeSize{-1}
    \LinesNumberedHidden
    \DontPrintSemicolon
    \Input{Binary string $y$ which is the output of the deletion adversary on $\text{ENC}(x)$.}
    \Output{A message ${x}'\in \Fq^k$.}

    \SetAlgoLined
    \nlset{0} Set $L$ to be the empty list. \;
    \nlset{1}
    \If{$y$ is a single run of zeros}{output $\bar{0} \in \Fq^k$ and return \;}
    \nlset{2}
    Every run of zeros of length at least $4\delin m$ is identified as an outer buffer. \;
    Let ${r}_1, \ldots, {r}_t$ be the strings between the outer buffers. \;
    \label{decode:outer-buffer}
    \nlset{3} \For{every ${r}_j$}{
    \label{decode:inner-part}
    {Every run of zeros of length at least $\delin m$ and less than $4 \delin m$ is identified as an inner buffer.}\;
    \If{exactly $1$ inner buffer was identified}
    {Denote by ${c}_j$ the string before the inner buffer and by ${c}_j'$ the string after the inner buffer. In particular ${r}_j = {c_j}\circ (\text{identified inner buffer}) \circ {c}_j'$.\footnote{Note that ${c}_j, {c}_j'$ are strings that start and end with the symbol $1$.} \; 
    \If{$m - 2\delin  m < |{c}_j| \leq m$ and $m - 2\delin  m < |{c}_j'| \leq m$}
    {$a = \textup{Dec}({c}_j)$ and $b = \textup{Dec}({c}_j')$. \;
    \If{$a$ is not $\perp$ and $b$ is not $\perp$}
    {Add to $L$ the tuple $(a, b)$.}}}
    }\label{decode:blow-up:end}
    \nlset{4} 
    {Decode $L$ using the algorithm of $\cCout$ given in \Cref{thm:half-linear}.} \label{decode:outer}\; 
    \caption{Decoding algorithm for \Cref{cnst:binary-lin}}
    \end{algorithm}

\subsection{Analysis}

\begin{prop}\label{prop:bin-decod}
    The code defined in \Cref{cnst:binary-lin} can correct from $\rho\delout mn$ adversarial deletions, using Algorithm~\ref{alg:decode}, in $O(n^3)$ time.
\end{prop}
\begin{proof}
    Let $x\in\Fq^{k}$ be a message and denote by $\sigma := \left( (\sigma_1, S_1\cdot \sigma_1), \ldots, (\sigma_n, S_n \cdot \sigma_n) \right)$, the outer codeword  corresponding to $x$, i.e., $\sigma=\cCout(x)$. We first note that if $x$ is the zero message then since the adversary is allowed to perform only deletions to ENC$(x)$, the input to the algorithm is a single run of zeros. Therefore the algorithm will output the zero message as required. Thus, from now on, we assume that $x$ is not the zero message. 
    
    We will upper bound the edit distance between $\sigma$ and $L$ that is obtained after performing step~\ref{decode:inner-part} of Algorithm~\ref{alg:decode}. If it holds that $\ed(\sigma, L) \leq \delout n$, then the decoding succeeds since our outer code, $\cCout$, can correct from $\delout n$ insdel errors. 
    
    Before we continue with the proof, we note that the outer codeword, $\sigma$, might have zero symbols (which are of the form $(0,0$)). Note that such a symbol is encoded, by the  inner code, to a long run of zeros, which is then  interpreted by our algorithm as an outer buffer. As can be seen in the proof of \Cref{thm:half-linear} (see  \Cref{rem:zeros-in-hlinear}), we only care about nonzero symbols. Namely, if we denote by $\sigma^{0}$ the string obtained from $\sigma$ by deleting all the zero symbols, then as long as $\ed(\sigma^{0}, L) < \delout n$, the decoding algorithm succeeds. Thus, we do not need to insert these zero symbols to $L$.

    Assume then that $x\neq 0$.
    In Step~\ref{decode:outer-buffer} the decoding algorithm identifies outer buffers.  We say that the algorithm identified correctly the $i$th outer buffer if in Step~\ref{decode:outer-buffer} it  identified an outer buffer  that contains one of the surviving symbols of the $i$th outer buffer of ENC$(x)$, and that contains no symbol of any other outer buffer of ENC$(x)$. We call such an identified outer buffer a \emph{genuine outer buffer}. 
    Observe, that if the $i$th outer symbol is $\sigma_i = (0,0)$, then the algorithm may identify the entire run between the $(i-1)$th and the $i$th outer buffers as a single outer buffer. In this case, too we say that this is a genuine outer buffer. The reason for that will become clear during the analysis. In a nutshell, the reason for not treating it as an erroneous buffer follows from the discussion above that shows that our algorithm ignores the zero outer symbol  (see \Cref{rem:zeros-in-hlinear}).
    In all other cases, we say that the decoder identified a \emph{fake outer buffer}.  We call an outer buffer that was not identified as an outer buffer (because the adversary deleted many $0$s from it) a \emph{corrupted outer buffer}.


        


    After identifying the outer buffers in  Step~\ref{decode:outer-buffer}, we get $t$ strings ${r}_1, \ldots, {r}_t$. 
    We  distinguish between three different types of ${r}_j$s, depending on the outer buffers that the algorithm identified:
    \begin{enumerate}[label=Type-\arabic* ${r}_j$, wide=0pt, leftmargin=*]
    \item  --  there exists an $i\in [n - 1]$ such that the algorithm identified the $(i-1)$th {genuine} outer buffer before ${r}_j$ and the $i$th {genuine} outer buffer after ${r}_j$. If $j = 1$ ($t$) then we require the algorithm to identify only the right (left) outer buffer.\label{type:1}
    \item -- if the buffers surrounding ${r}_j$ are genuine outer buffers that do not correspond to consecutive outer buffers in ENC$(x)$.\label{type:2}
    \item -- if at least one of the buffers surrounding ${r}_j$ is a fake outer buffer.\label{type:3}
    \end{enumerate}


    We first study how the adversary can create a \ref{type:1} that is not decoded correctly in Step~\ref{decode:inner-part}. In what follows, for a substring $s$ of ENC$(x)$, we denote with $\tilde{s}$ the remaining subsequence of $s$ after the deletions performed by the adversary. 

    \paragraph{\ref{type:1}:}
        In this case, ${r}_j$ is the form 
        \[
        {r}_j = \widetilde{ \cCin (\sigma_i)} \circ \widetilde{\inbuf}\circ  \widetilde{ \cCin (S_i \sigma_i)} \;,
        \]
        and we assume that the (original) $i$th buffer preceding ${r}_j$ and the $(i+1)$th buffer following ${r}_j$ were identified by the algorithm. 
        
        We  say that ${r}_j$ is a \emph{surviving outer symbol} if a single inner buffer was identified inside ${r}_j$ (thus ${r}_j = {c_j}\circ (\text{identified inner buffer}) \circ {c}_j'$),  and the decoding algorithm of the inner code returns $ \sigma_i$ and $S_i\cdot \sigma_i$ when given ${c}_j$ and   ${c}_j'$, respectively. 
        If in Step~\ref{decode:inner-part} the algorithm adds to $L$ the tuple $(a, b)\neq (\sigma_i, S_i\sigma_i)$, when going over ${r}_j$, then we call ${r}_j$ a \emph{fake outer symbol}. Note that the algorithm can also ignore ${r}_j$ in Step~\ref{decode:inner-part} and in this case, we call ${r}_j$ an \emph{ignored outer symbol}. For example, if ${r}_j$ contains several runs of zeros of length $\geq \delin m$, then several inner buffers are identified inside ${r}_j$, in which case the algorithm will not add anything to $L$.  

        Our objective is to show that the adversary has to perform at least $\rho m + 1$ deletions to $\cCin (\sigma_i)\circ \inbuf \circ \cCin (S_i \sigma_i)$ in order to create a \ref{type:1}  that gets ignored by our algorithm and at least $\delin m + \rho m$ deletions in order to create a \ref{type:1} that is a fake outer symbol.
        We say that the algorithm \emph{identified correctly the inner buffer} if exactly one inner buffer was identified inside ${r}_j$ and at least one of the bits in the identified inner buffer belongs to the original inner buffer.

The following claim shows that if the inner buffer was identified correctly and the adversary performed at most $\rho m$ deletions to each of the inner codewords, then  the decoding algorithm of the inner code successfully decodes ${c}_j$ and ${c}_j'$.

\begin{claim} \label{clm:inner-code-dec}
   Assume that the algorithm identified correctly the inner buffer inside ${r}_j$ (thus, ${r}_j = {c_j}\circ (\text{identified inner buffer}) \circ {c}_j'$). Then, as long as the adversary performed $\leq \rho m$ deletions to $\cCin (\sigma_i) $ ($\cCin (S_i \cdot \sigma_i)$), the decoding algorithm of the inner code, outputs correctly $\sigma_i$ ($S_i\cdot \sigma_i$) when given ${c}_j$ (${c}_j'$).
\end{claim}

\begin{proof}
First, note that it may be the case that a string of $0$s of an inner codeword (i.e., of  ${ \cCin (\sigma_i)}$ or of $ { \cCin (S_i \sigma_i)}$) are identified as a part of the inner or outer buffers. This is because our algorithm identifies buffers whenever it encounters a long enough run of zeros. Therefore, if $\cCin(\sigma_i)$ starts with a run of zeros, then this run is identified by our algorithm as part of the first outer buffer. The same phenomenon happens if $\cCin(\sigma_i)$ ends with a run of zeros, only this time the zeroes are identified as part of the inner buffer. 
Denote by $\cCin (\sigma_i)'$ the substring of $\cCin (\sigma_i)$ obtained by deleting the first and last run of zeros. By \Cref{prop:inner-code}\eqref{prop:inner-2}, $\cCin (\sigma_i)'$ is of length $\geq (1 - 2(\delin - \rho))m$. 

Note that the adversary has the option to delete $1$s from the beginning (or end) of $\cCin (\sigma_i)'$ and as a result, further $0$s will be identified as part of a buffer by the algorithm. For example, assume $11010010$ to be the first eight bits of $ \cCin (\sigma_i)'$ and further assume that the adversary deletes the first three $1$s from the left. In this case, we have ${\color{red}11}{\color{blue}0}{\color{red}1}{\color{blue}00}10$, where the red $1$s were deleted by the adversary and the blue $0$s are interpreted, by the algorithm, as part of the left outer buffer. 
Denote by $b_1$ the number of consecutive $1$s deleted from the beginning of $\cCin (\sigma_i)$ and by $e_1$ the number of consecutive $1$s deleted from the end of $\cCin (\sigma_i)$ where $b_1+e_1 \leq \rho m$, then, the number of zeros merged to the buffer is at most 
\[
\ceil{(b_1 + 1) / (\rho m + 1 )} (\delin m - \rho m) + \ceil{(e_1+1) / (\rho m + 1)}(\delin m - \rho m) = 2(\delin m - \rho m) \;.
\]
Denote the resulting string (after removing the first and last runs of $0$s that were created by the adversary after deleting $b_1+e_1$ $1$s) by $\cCin (\sigma_i)''$ and note that  $\cCin (\sigma_i)''$ is a substring of $\cCin (\sigma_i)$ of length $\geq (1 - 2\delin + \rho )m$. Now, the adversary can perform another $ \rho m -(b_1+e_1)$ deletions to the rest of the bits of $\cCin (\sigma_i)''$. In total, LCS$({c}_j,  \cCin(\sigma_i)'') \geq |\cCin(\sigma_i)''|-\rho m$. \Cref{prop:inner-code}\eqref{prop:inner-1} guarantees that we decode this corrupted codeword successfully.
\end{proof}

Thus, in order for the adversary to make the algorithm ignore ${r}_j$ or interpret it as a fake outer symbol, it must either
\begin{enumerate}[label=Case\,\arabic*:, wide=0pt, leftmargin=*]
   \item delete enough $0$s so that no inner buffer is identified, or \label{case:no-inner}
   \item delete many $1$s so that more than one inner buffer is identified, or \label{case:many-inner}
   \item delete bits so that only a single inner buffer is identified, but that the decoding algorithm fails.\label{case:many-dele}
\end{enumerate}


We study each of these cases separately.\\

Analysis of \ref{case:no-inner} In this case, the adversary must have deleted at least $\delin m + 1$ bits from the original inner buffer. In this case, ${r}_j$ is ignored by the algorithm.\\

Analysis of \ref{case:many-inner}  In this case, the algorithm identifies (at least) two inner buffers in ${r}_j$, and as a result, ignores it. \Cref{prop:inner-code}\eqref{prop:inner-2} implies that the adversary must delete at least $\rho m + 1$ many $1$s from an inner codeword in order to create a second long run of $0$s that is interpreted as an inner buffer.\\

Analysis of \ref{case:many-dele} We now assume that the algorithm identified a single inner buffer. If this inner buffer does not contain any bit of the original inner buffer, then, by the two previous cases, the adversary must have deleted at least $\delin m + 1$ many $0$ from the original inner buffer and additionally at least $\rho m + 1$ many $1$s from an inner codeword. In total, at least $\delin m + \rho m +2$ many bits were deleted. In this case, either ${r}_j$ is ignored by the algorithm, or it becomes a fake outer symbol.

If  the algorithm correctly identified the inner buffer, then \Cref{clm:inner-code-dec} implies that, for the algorithm to fail to decode, the adversary must have deleted more than $\rho m$ bits inside  $\cCin (\sigma_i) $ or $\cCin (S_i \cdot \sigma_i)$. In particular, the adversary must perform more than $\rho m$ deletions for the decoding to fail. Notice that in this case, the decoding algorithm of the inner code will output $\perp$ and will not return a fake outer symbol.\\

To conclude, if the adversary wishes to create a \ref{type:1} that is an ignored outer symbol, it needs to perform at least $\rho m + 1$ deletions. In order to create a \ref{type:1} that is a fake outer symbol, the adversary needs to delete at least $\delin m + \rho m+2$ many bits.

Observe that  an ignored outer symbol increases  $\ed(\sigma^0, L)$ by $1$ since the corresponding outer symbol, $(\sigma_i, S_i\cdot \sigma_i)$, was not added to $L$. A \ref{type:1} that is a fake outer symbol increases $\ed(\sigma^0, L)$ by $2$ since instead of the original outer symbol, a fake outer symbol is added to $L$.  Thus, the number of deletions that the adversary has to ``pay'' in order to increase  $\ed(\sigma^0, L)$ by $1$, in the case of \ref{type:1}, is at least 
\[\min\left\{ \rho m + 1, \frac{\delin m + \rho m +2}{2} \right\} = \rho m + 1 \;,
\]
where the equality follows as $\delin > 2.5 \rho$. Thus, in the case of \ref{type:1}, it is more ``economical'' for the adversary to make the algorithm ignore it rather than make it a fake outer symbol.

\paragraph{\ref{type:2}:}
    In this case, we assume that ${r}_j$ is such that the outer buffer identified before ${r}_j$ and the outer buffer identified after ${r}_j$ are genuine but not consecutive (and there is no fake outer buffer in between). Assume that the  outer buffer before $r_j$ corresponds to the $i_1$th outer buffer in ENC$(x)$ and that the  outer buffer after $r_j$ corresponds to the $i_2$th original outer buffer. In particular, the $i_2-i_1-1$ outer buffers between the $i_1$th and $i_2$th  were corrupted by the adversary.

    We now consider how many deletions the adversary had to perform in order for the algorithm to return a fake outer symbol. Note that the substring of the original codeword that starts at the first $1$ following the $i_1$th outer buffer and ends at the last $1$ preceding the $i_2$th outer buffer is of length at least  
    \begin{align*}
     2((1-\delin +\rho)+ 2\delin  +1)m  &+ 5\delin m + (i_2-i_1-2)(2+7\delin )m\\ &=(i_2-i_1)(2+7\delin )m - (7\delin -2\rho )m  \;.
    \end{align*}
    Observe that for the algorithm to not ignore $r_j$ we must have that $|r_j|< 2m+ 4\delin m $. It follows that for the algorithm not to ignore $r_j$, the adversary must have deleted at least 
    \[(i_2-i_1)(2+7\delin )m - (7\delin -2\rho )m-(2+ 4\delin  )m=(i_2-i_1)(2+7\delin )m - (2  + 11\delin -2\rho)m
    \]
    many bits. Creating such a fake outer symbol increases $\ed(\sigma^0, L)$ by $i_2 - i_1+1$ as it corresponds to deleting the outer symbols in locations $i_1,\ldots,i_2-1$ and an insertion of the fake outer symbol.
    
    If the adversary only corrupted the outer buffers between the $i_1$th and the $i_2$th outer buffers, without creating a fake outer symbol, then it must have deleted at least $(i_2-i_1-1)(\delin m +1)$ many $0$s. Indeed, to corrupt a single outer buffer (at least) $\delin m +1$ many $0$s have to be deleted. Such a behaviour by the adversary increases $\ed(\sigma^0, L)$ by $i_2 - i_1$ as it is equivalent to deleting the outer symbols in locations $i_1,\ldots,i_2-1$.

     Thus, the number of deletions that the adversary has to ``pay'' in order to increase  $\ed(\sigma^0, L)$ by $1$, in the case of \ref{type:2}, is at least 
     \begin{align*}
     &\min\left\{\frac{(i_2-i_1)(2+7\delin )m - (2  + 11\delin -2\rho)m}{(i_2-i_1+1)} \, ,\, \frac{(i_2-i_1-1)(\delin m +1)}{(i_2-i_1)} \right\}  \\
     &=\frac{(i_2-i_1-1)(\delin m +1)}{(i_2-i_1)}\,.
     \end{align*}
    Observe that $\frac{(i_2-i_1-1)(\delin m +1)}{(i_2-i_1)}>\rho m + 1$ and hence the adversary has to make more deletions in the case of \ref{type:2} than in the case of \ref{type:1} in order to increase  $\ed(\sigma^0, L)$ by $1$.

 
        \paragraph{\ref{type:3}:} Let us assume without loss of generality that the outer buffer to the left of $r_j$ is a fake outer buffer. 
        
        To create a fake outer buffer, the adversary has to create a run of $0$s of length $\geq4 \delin m$ such that all the bits in this run do not belong to any outer buffer in ENC$(x)$ (or that belong to two different outer buffers in ENC$(x)$. We treat this case later). The adversary faces two options; it can either merge many $0$s to an inner buffer or create a run of $0$s of length $\geq 4\delin m$ inside an inner codeword. By \Cref{prop:inner-code}\eqref{prop:inner-2}, the second case requires at least $4\rho m+4$ many deletions.
        In the first case,  the adversary needs to merge $\geq 2 \delin m$ many $0$s to an inner buffer. We claim that in this case, it must delete more than $\rho m + 1$ many $1$s from the inner codewords. Indeed, by \Cref{prop:inner-code}\eqref{prop:inner-2}, any $\delin m$ coordinates of an inner codeword contain at least $\rho m + 1$ many $1$s. As at least $\delin m$ $0$s must come from either the inner codeword to the left of the inner buffer or from the one to the right of the inner buffer, the claim follows.
        
        Now that we know the ``cost'' of creating a fake outer buffer, we shall analyze several cases. Denote with $i_1$ the index such that the last bit of the fake outer buffer came from the encoding of $(\sigma_{i_1},S_{i_1}\cdot \sigma_{i_1})$. 
        \begin{enumerate}
            \item The outer buffer to the right of $r_j$ is a genuine outer buffer corresponding to the $({i_1})$th outer buffer in ENC$(x)$: In this case it is not hard to verify that $|c_j|+|c_j'|< 2m-4\delin m$ and $r_j$ gets ignored. This increases $\ed(\sigma^0, L)$ by $1$, and, by the analysis above, the adversary had to make at least $\rho m+1$ many deletions.
            
            \item The outer buffer to the right of $r_j$ is a genuine outer buffer, but not the $i_1$th one: 
            Let us assume 
            that the genuine outer buffer to the right of $r_j$ is the $i_2$th outer buffer (observe that we must have $i_2>i_1$). We now consider two subcases:
            \begin{enumerate}
                \item The algorithm ignored $r_j$: As all the outer buffers between the $i_1$th and the $i_2$th were corrupted, the adversary must have deleted at least $(i_2-i_1)(\delin m+1)$ many $0$s. This increases  $\ed(\sigma^0, L)$ by at most $i_2-i_1+1$ as it causes the deletion of all symbols in locations $i_1+1,\ldots,i_2$, and potentially also the $i_1$th symbol. Thus, the average cost of increasing the edit distance by $1$ in this case is at least $\frac{(i_2-i_1)(\delin m+1)}{i_2-i_1+1}> (\delin m+1)/2 > \rho m +1$.
                
                \item The algorithm decoded $r_j$ to a fake outer symbol: Similarly to the analysis of \ref{type:2}, we see that  in this case, as the algorithm has to identify a single inner buffer inside $r_j$, and the length of $r_j$ is $|r_j|\leq 2m+4\delin m$, the adversary must have deleted at least 
                \[(i_2-i_1)(7\delin +2)m - (2+4\delin)m =(i_2-i_1-1)(7\delin +2)m+3\delin m\]
                many bits. This increases $\ed(\sigma^0,L)$ by at most $i_2-i_1+2$ since (as in the previous case) this caused at most $i_2-i_1+1$ many deletions and a single insertion. Thus, the average cost of increasing the edit distance by $1$ in this case is at least 
                $\frac{(i_2-i_1-1)(7\delin +2)m+3\delin m}{i_2-i_1+2}\geq \delin m > \rho m +1$.
            \end{enumerate}
            
            \item   The outer buffer to the right of $r_j$ is also a fake outer buffer: Let us assume  that the outer buffer to the right or $r_j$ was created inside the encoding of the $i_2$th outer symbol.
               We analyze two cases:

            \begin{enumerate}
                \item $i_2=i_1$: In this case, it is not hard to see that $r_j$ is too short and hence gets ignored by the algorithm. This increases $\ed(\sigma^0, L)$ by $1$. Note that by the analysis above, the adversary had to make at least  $(4\rho m + 4)+(\rho m +1)=5\rho m +5$ many deletions.
                
                \item $i_2>i_1$: Similar calculations as in the case of \ref{type:2} show that in this case, the adversary has to make more than $\rho m +1$ many deletions in order to increase  $\ed(\sigma^0, L)$ by $1$.  Indeed, let us assume that the first bit in the second fake outer buffer came from $(\sigma_{i_2},S_{i_2}\cdot \sigma_{i_2})$. It follows that in order to corrupt all the outer buffers between the $i_1$th and the $(i_2-1)$th outer buffers, the adversary must delete at least $(i_2-i_1)\delin m$ many bits. In this case, if $r_j$ is not interpreted as a fake outer symbol, then
                $\ed(\sigma^0,L)$ grew by at most $i_2-i_1+1$. If $r_j$ was decoded to a fake outer symbol, then we note that it must be the case that at most one inner buffer was identified inside $r_j$. Thus, at least $(i_2-i_1-1)4\delin m$ many more bits had to be deleted. In addition, we recall that at least $\rho m +1$ deletions occurred to create the outer buffer to the left of $r_j$ (we do not charge anything for the right one in order to avoid double-counting). Calculating, we see that the average cost of increasing the edit distance by $1$ in either of the cases is larger than $\rho m +1$.
                \end{enumerate}
                    \end{enumerate}

        Finally, we note that if the fake outer buffer before $r_j$ contains bits from two different original outer buffers, the $i_1$th and the $i_2$th, then at least  $2(i_2-i_1)\frac{(1-\delta+\rho)\rho}{\delta}m$ many $1$s had to be deleted. Such an operation increases the edit distance by at most  $i_2-i_1$. In addition, we have to repeat the analysis above and take into consideration the cost of creating the buffer to the right of $r_j$, and the additional effect of $r_j$ on the edit distance (i.e., whether $r_j$ was ignored or decoded as a fake outer symbol, etc.). It is clear that in this case, the cost of increasing the edit distance by $1$ is much larger than $\rho m +1$.  \\

        In conclusion, in all cases, in order to increase  $\ed(\sigma^0, L)$ by $1$, the adversary has to make at least $\rho m +1$ many deletions. 
        Since the adversary can make at most $\delout \rho mn$ deletions, it follows that $\ed(\sigma^0, L) < \delout n$. Hence, by the assumption on the outer code, Step~\ref{decode:outer} of Algorithm~\ref{alg:decode} returns the correct message. This completes the correctness part of \Cref{prop:bin-decod}. All that is left is to analyze the running time complexity of the algorithm.

        \paragraph{Running time:} The claim about the running time follows by first noting that Step~\ref{decode:outer-buffer}, in which we identify the outer buffers, runs in linear time. Secondly, for each $r_j$, the run time of Step~\ref{decode:inner-part} is determined by the cost of the brute force decoding algorithm. This algorithm runs in exponential time in $m$, where $m = \poly(1/\epsout)$. Hence, Step~\ref{decode:inner-part} runs in time $n\cdot \poly(1/\epsout)$. Finally, according to \Cref{thm:half-linear}, the decoding algorithm of the outer code runs in time $O(n^3)$. In conclusion, the running time of the decoding algorithm is $O(n^3)$. This concludes the proof of \Cref{prop:bin-decod}.

        \subsection{Proof of \Cref{thm:bin-linear}}

        \Cref{prop:bin-decod} implies that the code constructed in \Cref{cnst:binary-lin} can decode from $\rho \delout mn$ many deletions.  By Equation~\eqref{eq:rate-bin}, its rate is $\frac{\rRin \rRout}{2 + 7\delin}$.
        
          Recall that $\epsout < \delout/1400$,  $\delin = 1/6$, $\rho = 1/17$,  $\rRout = (1 - \delout)/4 - \epsout$ and $\rRin = \delin/16$. It follows that the rate of our code is 
    \begin{align*}
    {\mathcal R}&\geq \frac{\rRin \rRout}{2 + 7\delin} \\
&= \frac{1}{304} \cdot \left( \frac{1 - \delout}{4} - \epsout \right) \\
&\geq \frac{1}{304} \cdot \left( \frac{1 - 1.0029\delout}{4}\right)
\;,
\end{align*} 
and it can correct from more than $\delta = \delout \rho/(2 + 7\delin) > \delout/53.84$ fraction of worst-case deletions. Thus, we conclude that the final rate-error trade-off is
\[
\mathcal{R} \geq \frac{1 - 54 \cdot \delta}{1216} \;.
\]
\end{proof}

	\section{Open questions}
		In this paper, we studied linear codes that can handle insdel errors.
		Our main goal is, naturally,  to construct codes that get close (or match) the half-Singleton bound. 
		Over small alphabets, we constructed efficient linear codes that have relatively high rate compared to previous constructions. 
		We still do not have explicit constructions of linear codes over small fields that achieve the half-Singleton bound. As far as we know, \Cref{thm:full-linear} is the best explicit and efficient construction of linear insdel code over small fields. 
		Thus, the main open question is to construct efficient linear codes that match or get closer to the half-Singleton bound.
	\bibliographystyle{alpha}
	\bibliography{linearinsdel}

\newcommand{\etalchar}[1]{$^{#1}$}
\begin{thebibliography}{AGFC07}

\bibitem[AGFC07]{abdel2007linear}
Khaled~AS Abdel-Ghaffar, Hendrik~C Ferreira, and Ling Cheng.
\newblock On linear and cyclic codes for correcting deletions.
\newblock In {\em 2007 IEEE International Symposium on Information Theory},
  pages 851--855. IEEE, 2007.

\bibitem[BGZ17]{brakensiek2017efficient}
Joshua Brakensiek, Venkatesan Guruswami, and Samuel Zbarsky.
\newblock Efficient low-redundancy codes for correcting multiple deletions.
\newblock {\em IEEE Transactions on Information Theory}, 64(5):3403--3410,
  2017.

\bibitem[BLC{\etalchar{+}}16]{bornholt2016dna}
James Bornholt, Randolph Lopez, Douglas~M Carmean, Luis Ceze, Georg Seelig, and
  Karin Strauss.
\newblock {A DNA-based archival storage system}.
\newblock {\em ACM SIGARCH Computer Architecture News}, 44(2):637--649, 2016.

\bibitem[CGHL21]{cheng2020efficient}
Kuan Cheng, Venkatesan Guruswami, Bernhard Haeupler, and Xin Li.
\newblock Efficient linear and affine codes for correcting
  insertions/deletions.
\newblock In D{\'{a}}niel Marx, editor, {\em Proceedings of the 2021 {ACM-SIAM}
  Symposium on Discrete Algorithms, {SODA} 2021, Virtual Conference, January 10
  - 13, 2021}, pages 1--20. {SIAM}, 2021.

\bibitem[Che18]{cheraghchi2018capacity}
Mahdi Cheraghchi.
\newblock Capacity upper bounds for deletion-type channels.
\newblock In {\em Proceedings of the 50th Annual ACM SIGACT Symposium on Theory
  of Computing}, pages 493--506. ACM, 2018.

\bibitem[CHL{\etalchar{+}}19]{DBLP:conf/soda/ChengHLSW19}
Kuan Cheng, Bernhard Haeupler, Xin Li, Amirbehshad Shahrasbi, and Ke~Wu.
\newblock Synchronization strings: Highly efficient deterministic constructions
  over small alphabets.
\newblock In Timothy~M. Chan, editor, {\em Proceedings of the Thirtieth Annual
  {ACM-SIAM} Symposium on Discrete Algorithms, {SODA} 2019, San Diego,
  California, USA, January 6-9, 2019}, pages 2185--2204. {SIAM}, 2019.

\bibitem[CJLW18]{cheng2018deterministic}
Kuan Cheng, Zhengzhong Jin, Xin Li, and Ke~Wu.
\newblock Deterministic document exchange protocols, and almost optimal binary
  codes for edit errors.
\newblock In {\em 2018 IEEE 59th Annual Symposium on Foundations of Computer
  Science (FOCS)}, pages 200--211. IEEE, 2018.

\bibitem[CR03]{crochemore2003jewels}
Maxime Crochemore and Wojciech Rytter.
\newblock {\em Jewels of stringology: text algorithms}.
\newblock World Scientific, 2003.

\bibitem[CS20]{con2020explicit}
Roni Con and Amir Shpilka.
\newblock Explicit and efficient constructions of coding schemes for the binary
  deletion channel.
\newblock In {\em 2020 IEEE International Symposium on Information Theory
  (ISIT)}, pages 84--89. IEEE, 2020.

\bibitem[Dal11]{dalai2011new}
Marco Dalai.
\newblock A new bound on the capacity of the binary deletion channel with high
  deletion probabilities.
\newblock In {\em Information Theory Proceedings (ISIT), 2011 IEEE
  International Symposium on}, pages 499--502. IEEE, 2011.

\bibitem[GH21]{guruswami2021explicit}
Venkatesan Guruswami and Johan H{\aa}stad.
\newblock Explicit two-deletion codes with redundancy matching the existential
  bound.
\newblock In {\em Proceedings of the 2021 ACM-SIAM Symposium on Discrete
  Algorithms (SODA)}, pages 21--32. SIAM, 2021.

\bibitem[GHL21]{guruswami2021zero}
Venkatesan Guruswami, Xiaoyu He, and Ray Li.
\newblock The zero-rate threshold for adversarial bit-deletions is less than
  1/2.
\newblock {\em arXiv preprint arXiv:2106.05250}, 2021.

\bibitem[GK]{GuruswamiKuan21}
Venkatesan Guruswami and Cheng Kuan.
\newblock personal communication.

\bibitem[GW17]{guruswami2017deletion}
Venkatesan Guruswami and Carol Wang.
\newblock Deletion codes in the high-noise and high-rate regimes.
\newblock {\em IEEE Transactions on Information Theory}, 63(4):1961--1970,
  2017.

\bibitem[Hae19]{haeupler2019optimal}
Bernhard Haeupler.
\newblock Optimal document exchange and new codes for insertions and deletions.
\newblock In {\em 2019 IEEE 60th Annual Symposium on Foundations of Computer
  Science (FOCS)}, pages 334--347. IEEE, 2019.

\bibitem[Ham50]{hamming1950error}
Richard~W. Hamming.
\newblock Error detecting and error correcting codes.
\newblock {\em Bell System technical journal}, 29(2):147--160, 1950.

\bibitem[HMG19]{heckel2019characterization}
Reinhard Heckel, Gediminas Mikutis, and Robert~N Grass.
\newblock {A characterization of the DNA data storage channel}.
\newblock {\em Scientific reports}, 9(1):1--12, 2019.

\bibitem[HS17]{haeupler2017synchronization}
Bernhard Haeupler and Amirbehshad Shahrasbi.
\newblock {Synchronization strings: codes for insertions and deletions
  approaching the Singleton bound}.
\newblock In {\em Proceedings of the 49th Annual ACM SIGACT Symposium on Theory
  of Computing}, pages 33--46. ACM, 2017.

\bibitem[HS21]{haeupler2021synchronization}
Bernhard Haeupler and Amirbehshad Shahrasbi.
\newblock Synchronization strings and codes for insertions and deletions - {A}
  survey.
\newblock {\em {IEEE} Trans. Inf. Theory}, 67(6):3190--3206, 2021.

\bibitem[Kot96]{kotter1996fast}
Ralf Kotter.
\newblock {Fast generalized minimum-distance decoding of algebraic-geometry and
  Reed-Solomon codes}.
\newblock {\em IEEE Transactions on Information Theory}, 42(3):721--737, 1996.

\bibitem[MBT10]{mercier2010survey}
Hugues Mercier, Vijay~K Bhargava, and Vahid Tarokh.
\newblock A survey of error-correcting codes for channels with symbol
  synchronization errors.
\newblock {\em IEEE Communications Surveys \& Tutorials}, 12(1):87--96, 2010.

\bibitem[MD06]{mitzenmacher2006simple}
Michael Mitzenmacher and Eleni Drinea.
\newblock A simple lower bound for the capacity of the deletion channel.
\newblock {\em IEEE Transactions on Information Theory}, 52(10):4657--4660,
  2006.

\bibitem[Mit09]{mitzenmacher2009survey}
Michael Mitzenmacher.
\newblock A survey of results for deletion channels and related synchronization
  channels.
\newblock {\em Probability Surveys}, 6:1--33, 2009.

\bibitem[Sha48]{shannon1948mathematical}
Claude~Elwood Shannon.
\newblock A mathematical theory of communication.
\newblock {\em Bell system technical journal}, 27(3):379--423, 1948.

\bibitem[Sti09]{stichtenoth2009algebraic}
Henning Stichtenoth.
\newblock {\em Algebraic function fields and codes}, volume 254.
\newblock Springer Science \& Business Media, 2009.

\bibitem[SV90]{skorobogatov1990decoding}
Alexei~N Skorobogatov and Serge~G Vladut.
\newblock On the decoding of algebraic-geometric codes.
\newblock {\em IEEE Transactions on Information Theory}, 36(5):1051--1060,
  1990.

\bibitem[TVZ82]{tsfasman1982modular}
Michael~A Tsfasman, Serge Vl{\u{a}}dutx, and Thomas Zink.
\newblock {Modular curves, Shimura curves, and Goppa codes, better than
  Varshamov-Gilbert bound}.
\newblock {\em Mathematische Nachrichten}, 109(1):21--28, 1982.

\bibitem[YGM17]{yazdi2017portable}
S.M. Hossein~Tabatabaei Yazdi, Ryan Gabrys, and Olgica Milenkovic.
\newblock {Portable and error-free DNA-based data storage}.
\newblock {\em Scientific reports}, 7(1):1--6, 2017.

\end{thebibliography}

\end{document}